\documentclass[envcountsame]{llncs}

\usepackage[T1]{fontenc}

\usepackage{booktabs} 
\usepackage{diagbox}	
\usepackage{makecell} 
\usepackage{textcomp} 

\usepackage{amsmath}
\usepackage{amssymb}
\usepackage{stmaryrd}
\usepackage{centernot}
\usepackage[noadjust, nocompress]{cite}
\usepackage[retainorgcmds]{IEEEtrantools}
\usepackage{caption}

\usepackage[textwidth=3cm, shadow, obeyFinal]{todonotes}

\DeclareMathOperator{\eventually}{\mathbf{F}}
\DeclareMathOperator{\globally}{\mathbf{G}}
\DeclareMathOperator{\weventually}{\mathbf{\overline{F}}}
\DeclareMathOperator{\wglobally}{\mathbf{\overline{G}}}

\DeclareMathOperator{\history}{\vartriangleleft}
\DeclareMathOperator{\nextx}{\mathbf{X}}

\newcommand{\otoprule}{\midrule[\heavyrulewidth]}

\newcommand*{\until}{\mathbin{\mathbf{U}}}

\newcommand*{\since}{\mathbin{\mathbf{S}}}

\newcommand*{\ltl}{\textup{\textmd{\textsf{LTL}}}}
\newcommand*{\ltlhistory}{$\textup{\textmd{\textsf{LTL}}}_{\history}$}

\newcommand*{\mitl}{\textup{\textmd{\textsf{MITL}}}}

\newcommand*{\mtl}{\textup{\textmd{\textsf{MTL}}}}

\newcommand*{\msoone}{\textup{\textmd{\textsf{MSO[$<, +1$]}}}}

\newcommand*{\hmtl}{\textup{\textmd{\textsf{HyperMTL}}}}
\newcommand*{\hltl}{\textup{\textmd{\textsf{HyperLTL}}}}

\newcommand*{\hpctl}{$\textup{\textmd{\textsf{HyperPCTL}}}$}
\newcommand*{\hstl}{\textup{\textmd{\textsf{HyperSTL}}}}
\newcommand*{\qptl}{\textup{\textmd{\textsf{QPTL}}}}

\newcommand*{\ta}{\textup{\textmd{\textsf{TA}}}}
\newcommand*{\mia}{\textup{\textmd{\textsf{MIA}}}}

\newcommand*\A{\mathcal A} 
\newcommand*\B{\mathcal B} 

\newcommand*\transitions{\Delta}

\newcommand*\R{\mathbb R}

\newcommand*\N{\mathbb N}
\newcommand*\sem[1]{\ensuremath{\llbracket#1\rrbracket}}

\renewcommand*\phi{\varphi}
\newcommand*\AP{\textup{\textmd{\textsf{AP}}}}

\usetikzlibrary{positioning,automata,arrows,hobby,decorations.pathreplacing,fit}

\tikzstyle{every state}=[minimum size=1.5em]

\tikzset{>=latex,auto,node distance=2.5cm, every
  loop/.style={looseness=6}, initial text={}, inner sep=0.5mm,
  loopright/.style={loop,looseness=6,out=35, in=-35},
  loopleft/.style={loop,looseness=6,out=145, in=215},
  loopabove/.style={loop,looseness=6,out=125, in=55},
  loopbelow/.style={loop,looseness=6,out=-125, in=-55}, }

\makeatletter
\renewcommand{\paragraph}{\@startsection{paragraph}{6}{\z@}{2ex}{-0.7em}{\normalsize\bf}}
\makeatother

\let\doendproof\endproof
\renewcommand\endproof{~\hfill\qed\doendproof}

\newcommand{\qedhere}{%
  \begingroup \let\mathqed\math@qedhere
    \let\qed@elt\setQED@elt \QED@stack\relax\relax \endgroup
}

\title{On Verifying Timed Hyperproperties\thanks{This work was
supported by EPSRC grants EP/K026399/1 and EP/P020011/1.}}

\author{Hsi-Ming Ho, Ruoyu Zhou, and Timothy M. Jones}

\institute{
  University of Cambridge, United Kingdom
}

\begin{document}

\maketitle
\thispagestyle{plain}

\begin{abstract}
We study the satisfiability and model-checking problems for \emph{timed hyperproperties} specified with \hmtl{}, a timed extension of \hltl{}. Depending on whether interleaving of events in different traces is allowed, two possible semantics can be defined for timed hyperproperties: \emph{synchronous} and \emph{asynchronous}. While the satisfiability problem can be decided similarly as for \hltl{} regardless of the choice of semantics, we show that the model-checking problem, unless the specification is alternation-free, is undecidable even when very restricted timing constraints are allowed. On the positive side, we show that model checking \hmtl{} with quantifier alternations is possible under certain conditions
in the synchronous semantics, or when there is a fixed bound on the length of the time domain.
\end{abstract}

\section{Introduction} 

\paragraph{Background.}
One of the most popular specification formalisms for reactive systems
is \emph{Linear Temporal Logic} (\ltl{}), first introduced into computer science
by Pnueli~\cite{Pnueli1977} in the late 1970s.
The success of \ltl{} can be attributed 
to the fact that its satisfiability and model-checking problems are
of lower complexity ($\mathrm{PSPACE}$-complete, as compared with
non-elementary for the equally expressive first-order logic of order) and it enjoys simple
translations into automata and excellent tool support
(e.g.,~\cite{Holzmann1997, Cimatti2002}).

While \ltl{} is adequate for describing features of \emph{individual} execution traces,  
many security policies in practice are based on relations between \emph{two (or more)} execution traces.
A standard example of such properties is \emph{observational determinism}~\cite{Roscoe95, ZdancewicM03, HuismanWS06}: 
for every pair of execution traces, if the low-security inputs agree
in both execution traces, then the low-security outputs in both execution traces must agree as well.
Such properties are called \emph{hyperproperties}~\cite{ClarksonS10}: a model of the property is not
a single execution trace but a set of execution traces.
\hltl{}~\cite{ClarksonFKMRS14}, obtained from \ltl{} by adding \emph{trace quantifiers}, 
has been proposed as a specification formalism to express hyperproperties.
For example, operational determinism can be expressed as the \hltl{} formula:
\[
\forall \pi_a \, \forall \pi_b \, \wglobally (I_a = I_b) \Rightarrow \wglobally (O_a = O_b) \,.
\]
\hltl{} inherits almost all the benefits of \ltl{}; in particular, tools that support \hltl{} verification
can be built by leveraging existing tools for \ltl{}.

For many applications, however, in addition to the occurences and orders of events,
\emph{timing} has to be accounted as well.
For example, one may want to verify that in every execution trace of the system, whenever
a request $\mathtt{req}$ is issued, the corresponding acknowledgement $\mathtt{ack}$
is received within the next $5$ time units.
\emph{Timed automata}~\cite{AluDil94} and \emph{timed logics}~\cite{Koy90, AluHen94, AluFed96}
are introduced exactly for this purpose.
In the context of security, timing anomalies caused by different high-security inputs
is a realistic attack vector that can be exploited to obtain sensitive information;
this kind of \emph{timing side-channel attacks} also play significant roles in
high-profile exploits like Meltdown~\cite{Lipp0G0HFHMKGYH18} and Spectre~\cite{Kocher2018}.
In order to detect such undesired characteristics of systems, one needs to be able to reason about
\emph{timed hyperproperties}.
\begin{example}[\cite{LiSGFTO10}]
An AND gate with two inputs $A$, $B$ and an output $C$ and respective delays
$T_A$, $T_B$, and $T_C$ can be modelled as the timed automaton below (suppose that $T_A < T_B$
and $T_B - T_A < T_C$):
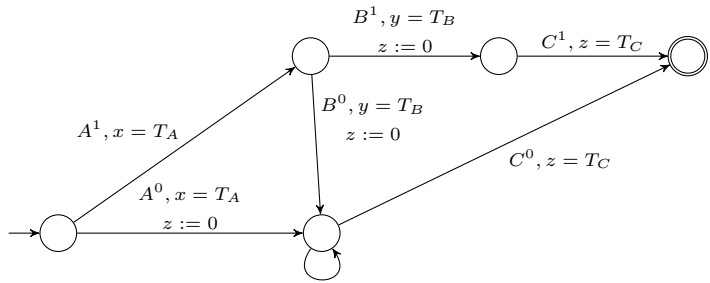
\vspace{-0.5cm}
\begin{figure}[h]
\centering
\scalebox{1}{
\begin{tikzpicture}[auto, >=stealth']
		\node[initial left ,state] (0) {};
		\node[state, above right= 2cm and 3cm of 0] (1) {};
		\node[state, right=2cm of 1] (2) {};
		\node[state, accepting, right=2cm of 2] (3) {};
		\node[state, right=3cm of 0] (4) {};
		\path
		(0) edge[->] node{\scriptsize $A^1, x = T_A$} (1)
		(1) edge[->] node[align=center]{\scriptsize $B^1, y = T_B$ \\ \scriptsize $z := 0$} (2)
		(1) edge[->] node[align=center]{\scriptsize $B^0, y = T_B$ \\ \scriptsize $z := 0$} (4)
		(2) edge[->] node{\scriptsize $C^1, z = T_C$} (3)
		(0) edge[->] node[align=center]{\scriptsize $A^0, x = T_A$ \\ \scriptsize $z := 0$} (4)
		(4) edge[loopbelow, ->] (4)
		(4) edge[->] node[swap]{\scriptsize $C^0, z = T_C$} (3);
\end{tikzpicture}
}
\end{figure}
\vspace{-0.5cm}

\noindent Of course, once $A$ turned out to be $0$ (i.e.~$A^0$ has happened), the output $C$ must be $0$ as well.
But the time when $C_0$ happens (assuming $C = 0$) also depends on the value of $A$;
in other words, in addition to the value of $C$, a low-security user can, when $C = 0$, also infer the value of $A$ (while she or he should not be able to).
In this simple example, however, the timing side channel can be removed
by adding $z := 0$ on the self-loop on the lower-right location.

\end{example}


\paragraph{Contributions.}
We propose \hmtl{},
obtained by adding trace quantifiers to \mtl{}, 
as a specification formalism for timed hyperproperties.
We consider systems modelled as \emph{timed automata}, and thus system behaviours are
sequences of \emph{events} that happen at different instants in time;
this gives two possible pointwise semantics of \hmtl{}: \emph{asynchronous} and \emph{synchronous}.
We show that, as far as satisfiability is concerned, \hmtl{} 
is similar to \hltl{}, i.e.~satisfiability is decidable for fragments not containing $\forall \exists$,
regardless of which semantics is assumed.
However, in contrast with \hltl{} (whose model-checking problem is decidable),  
model checking \hmtl{} is undecidable if there is at least one quantifier alternation
in the specification, even when the timing constraints used in either the system or
the specification are very restricted. 
Still, the alternation-free fragment of \hmtl{}, which is arguably sufficient
to capture many timed hyperproperties of practical interest, has a decidable model-checking problem. 
Finally, we identify several subcases where \hmtl{} model checking is decidable
for larger fragments, such as
when the synchronous semantics is assumed, the model is untimed, and the specification belongs to a certain subclass of
one-clock timed automata, or when the time domain is bounded \emph{a priori} by some $N \in \mathbb{N}_{> 0}$.

\paragraph{Related work.}
Since the pioneering work of Clarkson and Schneider~\cite{ClarksonS10},
there have been great interest in specifying and verifying hyperproperties
in the past few years. The framework based on \hltl{}~\cite{ClarksonFKMRS14} is possibly
the most popular for this purpose, thanks to its
expressiveness, flexibility, and relative ease of implementation.
In addition to satisfiability~\cite{FinkbeinerHS17, FinkbeinerHH18} and model checking~\cite{ClarksonFKMRS14, FinkbeinerRS15}, 
tools for monitoring \hltl{} also exist~\cite{AgrawalB16, FinkbeinerHST17, FinkbeinerHST18}.

Our formulation of \hmtl{} is very closely related to 
\hstl{}~\cite{NguyenKJDJ17} originally proposed in the context of quality assurance
of \emph{cyber-physical systems}.
While~\cite{NguyenKJDJ17} focusses on testing,
we are concerned with the decidability of verification problems. 
On the other hand, the semantics of \hstl{} is defined over sets of 
continuous signals, i.e.~state-based; 
as noted in~\cite{NguyenKJDJ17}, however, the price to pay for the extra generality is
that implementing a model checker for \hstl{} is very difficult,
especially for systems modelled in proprietary frameworks (such as Simulink\textregistered{}).
Practical reasoning of \hmtl{}, by contrast, can be carried out easily with 
existing highly optimised timed automata verification back ends, e.g.,~\textsc{Uppaal}~\cite{LarPet97}.
Indeed, a prototype model checker based on \textsc{Uppaal} for the synchronous semantics of \hmtl{}
(with some restrictions) is reported in~\cite{Heinen2018}, although it does not consider the decidability
of verification problems. Another relevant work~\cite{GerkingSB18}, also based on \textsc{Uppaal},
checks noninterference in systems modelled as timed automata (similar to Example~\ref{ex:noninterference}; see below).
Their approach, however, is specifically tailored to noninterference and does not generalise.


It is also possible to extend hyperlogics in other quantitative dimensions orthogonal to time.
\hpctl{}~\cite{AbrahamB18} can express \emph{probabilisitic hyperproperties}, e.g.,~the
probability distribution of the low-security outputs are independent of the high-security inputs.
In~\cite{FinkbeinerHT18}, specialised algorithms are developed for verifying \emph{quantitative hyperproperties},
e.g.,~there is a bound on the number of traces with the same low-security inputs but different low-level outputs. 
The current paper is complementary to these works.

\section{Timed hyperproperties}


\paragraph{Timed words.}
A \emph{timed word} (or a \emph{trace}) over a finite alphabet $\Sigma$ is a finite sequence
of \emph{events} $(\sigma_1,\tau_1)\dots(\sigma_n, \tau_n) \in (\Sigma \times \R_{\geq 0})^\ast$
with $\tau_1 \dots \tau_n$ an increasing
sequence of non-negative real numbers (`\emph{timestamps}'), i.e.~$\tau_i < \tau_{i+1}$ for all $i$, $1 \leq i < n$.
For $t \in \R_{\geq 0}$ and a timed word $\rho = (\sigma_1,\tau_1)\dots(\sigma_n, \tau_n)$,
we write $t \in \rho$ iff $t = \tau_i$ for some $i$, $1 \leq i \leq n$.
We denote by $T\Sigma^\ast$ the set of all timed words over $\Sigma$.
A \emph{timed language} (or a \emph{trace property}) is a subset of $T\Sigma^\ast$.

\paragraph{Timed automata.}
Let $X$ be a finite set of \emph{clocks}
($\R_{\geq 0}$-valued variables).
A \emph{valuation} $v$ for $X$ maps each clock $x \in X$ to a value in $\R_{\geq 0}$.
The set $G(X)$ of \emph{clock constraints} (\emph{guards}) $g$ over $X$ is generated
by $g:= \top\mid g\land g \mid x\bowtie c$ where
${\bowtie}\in \{{\leq},{<},{\geq},{>}\}$, $x\in X$, and $c\in\N_{\geq 0}$.
The satisfaction of a guard $g$ by a valuation $v$ (written $v \models g$) is
defined in the usual way.
For $t\in\R_{\geq 0}$, we let $v +t$
be the valuation defined by $(v +t)(x) = v (x)+t$ for all $x\in
X$. For $\lambda \subseteq X$, we let $v [\lambda \leftarrow 0]$ be the valuation
defined by $(v[\lambda \leftarrow 0])(x) = 0$ if $x\in \lambda$, and
$(v[\lambda \leftarrow 0])(x) = v (x)$ otherwise.

A \emph{timed automaton} (\ta{}) over $\Sigma$ is a tuple
$\A = \langle \Sigma, S, s_0, X, \transitions, F \rangle$ where $S$ is a finite set of
locations, $s_0 \in S$ is the initial location, $X$ is a finite set of clocks,
$\transitions \subseteq S \times \Sigma \times G(X) \times
2^X \times S$ is the transition relation,
and $F$ is the set of accepting locations.
We say that $\A$ is \emph{deterministic} iff for each $s \in S$ and $\sigma \in \Sigma$ and
every pair of transitions $(s, \sigma, g^1, \lambda^1, s^1) \in \transitions$ and $(s, \sigma, g^2, \lambda^2, s^2) \in \transitions$, $g^1 \land g^2$ is not satisfiable.
A \emph{state} of $\A$ is a pair $(s, v)$
of a location $s \in S$ and a valuation $v$ for $X$.
A \emph{run} of $\A$ on a timed word
$(\sigma_1,\tau_1)\dots(\sigma_n, \tau_n) \in T\Sigma^\ast$ is a
sequence of states $(s_0,v_0)\dots(s_n,v_n)$ where (i)
$v_0(x)=0$ for all $x \in X$ and (ii) for each $i$, $0 \leq i < n$,
there is a transition $(s_i,\sigma_{i+1},g,\lambda,s_{i+1})$
such that 
$v_i +(\tau_{i+1}-\tau_i)\models g$ (let $\tau_0=0$) and
$v_{i+1} =\big(v_i +(\tau_{i+1}-\tau_i)\big)[\lambda \leftarrow 0]$. 
A run of $\A$ is \emph{accepting} iff
it ends in a state $(s, v)$ with $s \in F$. A timed word
is \emph{accepted} by $\A$ iff $\A$ has an accepting run on it.
We denote by $\sem{\A}$ the timed language of $\A$,
i.e.~the set of all timed words accepted by $\A$.
Two fundamental results on \ta{s} are that the \emph{emptiness} problem 
is decidable ($\mathrm{PSPACE}$-complete), but the \emph{universality} problem
is undecidable~\cite{AluDil94}.

\paragraph{Timed logics.}

The set of \mtl{} formulae over a finite set of atomic propositions $\AP$ 
are generated by 
\begin{displaymath}
  \psi := \top \mid p \mid \psi_1 \land \psi_2 \mid \neg\psi \mid \psi_1 \until_I \psi_2 \mid \psi_1 \since_I \psi_2
\end{displaymath}
where $p\in\AP$ 
and $I \subseteq \R_{\geq 0}$ is a \emph{non-singular} interval with endpoints in $\N_{\geq 0} \cup\{\infty\}$.\footnote{In the literature, this logic (with the requirement that constraining intervals must be non-singular) is usually referred to as \mitl{}~\cite{AluFed96}, but we simply call it \mtl{} in this paper for notational simplicity. Also note that our undecidability results carry over
to the fragment with only future operators.}
We omit the subscript $I$ when $I = [0, \infty)$ and sometimes write pseudo-arithmetic expressions
for constraining intervals, e.g.,~`$< 3$' for $[0, 3)$.
The other Boolean operators are defined as usual:
$\bot \equiv \neg\top$ and $\psi_1\lor\psi_2 \equiv \neg(\neg\psi_1\land\neg\psi_2)$.
We also define the dual temporal operators
$\psi_1 \widetilde{\until}_I \psi_2 \equiv \neg \big( (\neg \psi_1) \until_I (\neg \psi_2) \big)$
and $\psi_1 \widetilde{\since}_I \psi_2 \equiv \neg \big( (\neg \psi_1) \since_I (\neg \psi_2) \big)$.
Using these operators, every \mtl{} formula $\psi$ can be transformed into an \mtl{} formula $\mathit{nnf}(\psi)$ in \emph{negative normal form},
i.e.~$\neg$ is only applied to atomic propositions.
To ease the presentation, we will also use the usual shortcuts like
$\eventually_I\psi \equiv \top\until_I\psi$, $\globally_I\psi \equiv \neg\eventually_I\neg\psi$,
$\nextx \psi \equiv \bot \until \psi$, and `weak-future' variants of temporal operators, e.g.,~$\weventually \psi \equiv \psi \vee \eventually \psi$.
Given an \mtl{} formula $\psi$ over $\AP$, 
a timed word $\rho$ over $\Sigma_\AP = 2^\AP$, and $t \in \R_{\geq 0}$, we define the \mtl{} satisfaction relation $\models$ as follows:\footnote{The formulation of the pointwise semantics of \mtl{}
here deviates slightly from the standard one (cf.~\cite{AluHen93, OuaWor07})
to enable interleaving of events in different traces.}
\begin{itemize}
\item $(\rho, t)\models \top$ iff $t \in \rho$;
\item $(\rho, t)\models \bot$ iff $t \notin \rho$;
\item $(\rho, t)\models p$ iff $t \in \rho$ and $p\in \sigma_i$;
\item $(\rho, t)\models \neg p$ iff $t \in \rho$ and $p\notin \sigma_i$;
\item $(\rho, t)\models \psi_1\land \psi_2$ iff $(\rho,t)\models\psi_1$ and 
$(\rho,t)\models\psi_2$;
\item $(\rho, t)\models \psi_1\lor \psi_2$ iff $(\rho,t)\models\psi_1$ or 
$(\rho,t)\models\psi_2$;
\item $(\rho,t)\models \psi_1 \until_I \psi_2$ iff there exists
  $t' > t$ such that $t'-t \in I$, $(\rho, t') \models \top$, $(\rho, t') \models \psi_2$, and $(\rho, t'') \models \psi_1$ for all $t''$ such that
$t'' \in (t, t')$ and $(\rho, t'') \models \top$;
\item $(\rho,t)\models \psi_1 \widetilde{\until}_I \psi_2$ iff for all
  $t' > t$ such that $t'-t \in I$ and $(\rho, t') \models \top$, either $(\rho, t') \models \psi_2$ or $(\rho, t'') \models \psi_1$ for some $t''$ such that
$t'' \in (t, t')$ and $(\rho, t'') \models \top$;
\item $(\rho,t)\models \psi_1 \since_I \psi_2$ iff there exists
  $t'$, $0 \leq t' < t$ such that $t-t' \in I$, $(\rho, t') \models \top$, $(\rho, t') \models \psi_2$, and $(\rho, t'') \models \psi_1$ for all $t''$ such that
$t'' \in (t', t)$ and $(\rho, t'') \models \top$;
\item $(\rho,t)\models \psi_1 \widetilde{\since}_I \psi_2$ iff for all
  $t'$, $0 \leq t' < t$ such that $t-t' \in I$ and $(\rho, t') \models \top$, either $(\rho, t') \models \psi_2$ or $(\rho, t'') \models \psi_1$ for some $t''$ such that
$t'' \in (t', t)$ and $(\rho, t'') \models \top$;
\item $(\rho, t)\models \neg \psi$ iff $(\rho, t) \models \mathit{nnf}(\neg \psi)$.
\end{itemize}
We say that $\rho$ \emph{satisfies} $\psi$ ($\rho\models\psi$)
iff $(\rho,0)\models\psi$, and we write
$\sem\psi$ for the timed language of $\psi$, i.e.~the set of all timed words satisfying $\psi$. 
It is well known that any \mtl{} formula
can be translated into a \ta{} accepting the same timed language~\cite{Alur1992}; this implies that \mtl{} satisfiability is decidable ($\mathrm{EXPSPACE}$-complete).

\paragraph{Adding trace quantifiers.}
Let $V$ be an infinite supply of \emph{trace variables}, the set of \hmtl{} formulae over $\AP$ are generated by
\begin{IEEEeqnarray*}{rCl}
  \phi & := & \exists \pi \, \phi \mid \forall \pi \, \phi \mid \psi \\
  \psi & := & \top \mid \top_\pi \mid p_\pi \mid \psi_1 \land \psi_2 \mid \neg\psi \mid \psi_1 \until_I \psi_2 \mid \psi_1 \since_I \psi_2
\end{IEEEeqnarray*}
where $\pi \in V$, $p\in\AP$, and $I \subseteq \R_{\geq 0}$ is a non-singular interval with endpoints in $\N_{\geq 0} \cup\{\infty\}$ (to ease the notation, we will usually write, e.g.,~$p_a$ for $p_{\pi_a}$).
Without loss of generality we forbid the reuse of trace variables, i.e.~each trace quantifier
must use a fresh trace variable.
Syntatic sugar is defined as in \mtl{}, e.g.,~$\eventually_I\psi \equiv \top\until_I\psi$.
A \hmtl{} formula is \emph{closed} if it does not have
free occurrences of trace variables.
Following~\cite{FinkbeinerH16}, we refer to fragments of
\hmtl{} by their quantifier patterns, e.g.,~$\exists^\ast \forall^\ast$-\hmtl{}.
Finally, note that trace quantifiers can be added to \ta{s} in the same manner.

In contrast with \ta{s} and \mtl{} formulae, which define \emph{trace properties},
\hmtl{} formulae define \emph{(timed) hyperproperties}, i.e.~sets of trace properties.
Depending on whether one requires timestamps in quantified traces to match exactly (i.e.~all quantified traces
must \emph{synchronise}), two possible semantics can be defined accordingly.

\paragraph{Asynchronous semantics.}
A \emph{trace assignment} over $\Sigma$ is a partial mapping from
$V$ to $T\Sigma^\ast$. We write $\Pi_\emptyset$ for the empty trace assignment
and $\Pi[\pi \mapsto \rho]$ for the trace assignment
that maps $\pi$ to $\rho$ and $\pi'$ to $\Pi(\pi')$ for all $\pi' \neq \pi$.
Given a \hmtl{} formula $\phi$ over $\AP$, 
a trace set $T$ over $\Sigma_\AP$, a trace assignment $\Pi$ over $\Sigma_\AP$, and $t \in \R_{\geq 0}$,
we define the \hmtl{} \emph{asynchronous}
satisfaction relation $\models$ as follows (we omit the cases where the definitions are obvious):
\begin{itemize}
\item $(T, t)\models_\Pi \top$ iff $t \in \rho$ for some $\rho \in \mathit{range}(\Pi)$;
\item $(T, t)\models_\Pi \top_\pi$ iff $t \in \rho$ for $\rho = \Pi(\pi)$;
\item $(T, t)\models_\Pi p_\pi$ iff $t \in \rho$ for $\rho = \Pi(\pi)$ and $p \in \sigma_i$
for the event $(\sigma_i, t)$ in $\rho$;
\item $(T, t) \models_\Pi \exists \pi \, \phi$ iff there is a trace $\rho \in T$ such that $(T, t) \models_{\Pi[\pi \mapsto \rho]} \phi$;
\item $(T, t) \models_\Pi \forall \pi \, \phi$ iff for all traces $\rho \in T$, $(T, t) \models_{\Pi[\pi \mapsto \rho]} \phi$.
\end{itemize}
We say that $T$ \emph{satisfies} a closed \hmtl{} formula $\phi$ in the asynchronous semantics ($T \models \phi$)
iff $(T,0) \models_{\Pi_\emptyset} \phi$. 
\begin{example}[Noninterference in event-based systems~\cite{GogMes82}]\label{ex:noninterference}
A system operating on sequences of commands issued by different users 
can be modelled as a deterministic finite automaton $\mathcal{A}$
over $\Sigma = U \times C$ where $U$ is the set of users and
$C$ is the set of commands. Additionally, let $\mathsf{Obs}$ be the set of observations
and $\mathit{out}: S \times U \rightarrow \mathsf{Obs}$ be the observation function
for what can be observed at each location by each user.
Let there be a partition of $U$ into two disjoint sets of users $U_H \subseteq U$ and $U_L \subseteq U$.
\emph{Noninterference} requires that for each $w \in \Sigma^\ast$
where $w$ ends with a command issued by a user in $U_L$ and $\A$ reaches $s$ after reading $w$,
the subsequence $w'$ obtained by removing all the commands issued by the users in $U_H$
results in a location $s'$ such that the observation $\mathit{out}(s', u_L)$ of each user $u_L \in U_L$ is identical to $\mathit{out}(s, u_L)$.
For our purpose, we can combine $\A$ and $\mathit{out}$ (in the expected way) into an automaton $\A'$ over $\Sigma_\AP$
where $\AP = (U \times C) \uplus (U \times \mathsf{Obs})$ (atomic propositions
in $U \times \mathsf{Obs}$ reflect the observations
at the location that has just been entered).
Checking noninterference then amounts to model checking $\A'$ (whose locations are all accepting)
against the following \hmtl{} formula in the asynchronous semantics:
\[
\phi = \forall \pi_a \, \forall \pi_b \, \big(\wglobally(\top_b \Rightarrow \psi^L_b \wedge \psi^{=}_{U, C}) \wedge \wglobally (\top_a \wedge \bot_b \Rightarrow \psi^H_a) \Rightarrow \wglobally(\top_b \Rightarrow \psi^{=}_{\mathit{out}(U)})\big)
\]
(where $\psi^L_b$ asserts that the command in $\pi_b$ is issued by a user in $U_L$, 
$\psi^=_{U, C}$ says that the two synchronised commands in $\pi_a$ and $\pi_b$ agree on $U$ and $C$, etc.).
Compared with the state-based formulation in~\cite{ClarksonFKMRS14}, allowing interleaving of events
leads to a much simpler and clearer specification.
\end{example}

\paragraph{Synchronous semantics.} A less general semantics can be defined for \hmtl{} formulae
where each trace quantifier only ranges over traces that synchronise with the
traces in the current trace assignment.
For example, the second quantifier in $\exists \pi_a \, \exists \pi_b \, \psi$ requires $\pi_b$ to satisfy $(\pi_a, t) \models \top_a \Leftrightarrow (\pi_b, t) \models \top_b$ for all $t \in \R_{\geq 0}$.
The \hmtl{} \emph{synchronous} satisfaction relation $\models^\mathit{sync}$
can, in fact, be expressed in the asynchronous semantics
by explicitly requiring newly quantified traces to synchronise in the quantifier-free part of the formula.
More precisely, for a closed \hmtl{} formula
$\phi = \mathcal{Q} \, \phi'$
where $\mathcal{Q}$ denotes
a block of quantifiers of the same type (i.e.~all existential or
all universal) and $\phi'$ is a possibly open \hmtl{} formula,
and a set $V$ of trace variables, let 
(abusing notation slightly) $\mathit{sync}(\phi, V) = \mathcal{Q} \, \big(\wglobally(\bigwedge_{\pi \in \mathcal{Q} \cup V} \top_\pi) \wedge \mathit{sync}(\phi', \mathcal{Q} \cup V)\big)$ when $\mathcal{Q}$ are
existential,
$\mathit{sync}(\phi) = \mathcal{Q} \, \big(\wglobally(\bigwedge_{\pi \in \mathcal{Q} \cup V} \top_\pi) \Rightarrow \mathit{sync}(\phi', \mathcal{Q} \cup V)\big)$ when $\mathcal{Q}$ are
 universal,
 and $\mathit{sync}(\psi, V) = \psi$ when $\psi$ is quantifier-free.
The following lemma holds subject to rewriting
the formula into prenex normal form.
\begin{lemma}\label{lem:syncinasync}
For any trace set $T$ over $\Sigma_\AP$ and closed \hmtl{} formula $\phi$ over $\AP$, $T \models^\mathit{sync}  \phi$ iff $T \models \mathit{sync}(\phi, \emptyset)$.
\end{lemma}
While the synchronous semantics may seem quite restricted (intuitively, the chance that
two random traces in a timed system have exactly the same timestamps is certainly slim!), one can argue
that it already suffices for many applications if \emph{stuttering steps}
are allowed.
We will see later that for alternation-free \hmtl{},
the asynchronous semantics can `simulated' in the synchronous semantics
using a `weak inverse' of Lemma~\ref{lem:syncinasync}.

\paragraph{Satisfiability and model checking.}
Given a closed \hmtl{} formula $\phi$ over $\AP$, the \emph{satisfiability}
problem asks whether there is a \emph{non-empty} trace set $T \subseteq T\Sigma_\AP^\ast$
satisfying it, i.e.~$T \models \phi$ (or
$T \models^\mathit{sync} \phi$, if the synchronous
semantics is assumed).
Given a \ta{} $\A$ over $\Sigma_\AP$ and a closed \hmtl{} formula $\phi$ over $\AP$,
the \emph{model-checking} problem asks whether $\sem{\A} \models \phi$ (or $\sem{\A} \models^\mathit{sync} \phi$).
Our focus in this paper is on the \emph{decidability} of these problems,
as their complexity (when they are decidable) follow straightforwardly 
from standard results on \mtl{}~\cite{AluFed96} and \hltl{}~\cite{ClarksonFKMRS14, FinkbeinerH16}.

%
%

\section{Satisfiability}\label{sec:satisfiability}

To `emulate' interleaving of events (of a concurrent or distributed system, say) in a synchronous, state-based setting, it is natural and necessary to introduce
stuttering steps.
In the context of verification, it is often a desirable trait
for a temporal logic to be \emph{stutter-invariant}~\cite{Lamp83a, Kucera05},
so that it cannot be used to differentiate traces that are essentially `the same'.
As a simple attempt to reconcile the asynchronous and synchronous semantics
of \hmtl{}, we can make use of \emph{silent events} in the same spirit to 
enable synchronisation of interleaving traces
while preserving the semantics.
More precisely, let $\mathit{stutter}(\rho)$ for a trace $\rho \in T\Sigma_\AP^\ast$ be the maximal set of traces $\rho' \in T\Sigma_{\AP_\epsilon}^\ast$ 
($\AP_\epsilon = \AP \cup \{ p^\epsilon \}$) such that
\begin{itemize}
\item for every event $(\sigma_i, \tau_i)$ in $\rho'$, either $\sigma_i = \{ p^\epsilon \}$ or $p^\epsilon \notin \sigma_i$;
\item $\rho$ can be obtained from $\rho'$ by deleting all the $\{ p^\epsilon \}$-events.
\end{itemize}
This extends to trace sets $T \subseteq T\Sigma_\AP^\ast$ in the obvious way.
For a closed alternation-free \hmtl{} formula $\phi = \mathcal{Q} \psi$ over $\AP$, let $\mathit{stutter}(\phi) = \mathcal{Q} \psi''$ be 
the \hmtl{} formula over $\AP_\epsilon$ obtained by (recursively) replacing in $\psi$,
e.g.,~all subformulae $\psi_1 \until_I \psi_2$ with 
$(\bigvee_{\pi \in \mathcal{Q}} \neg p_\pi^\epsilon \Rightarrow \psi_1) \until_I (\bigvee_{\pi \in \mathcal{Q}} \neg p_\pi^\epsilon \wedge \psi_2)$, to give $\psi'$, and finally let $\psi'' = \wglobally(\bigvee_{\pi \in \mathcal{Q}} \neg p^\epsilon_\pi ) \wedge \big(\bigwedge_{\pi \in \mathcal{Q}} \wglobally(p^\epsilon_\pi \Rightarrow \bigwedge_{p \in \AP} \neg p_\pi)\big) \wedge \psi'$
when $\mathcal{Q}$ are existential and
$\psi'' = \wglobally(\bigvee_{\pi \in \mathcal{Q}} \neg p^\epsilon_\pi ) \wedge \big(\bigwedge_{\pi \in \mathcal{Q}} \wglobally(p^\epsilon_\pi \Rightarrow \bigwedge_{p \in \AP} \neg p_\pi)\big) \Rightarrow \psi'$
when $\mathcal{Q}$ are universal.
Intuitively, $\psi''$ ensures that the traces involved are \emph{well-formed} (i.e.~satisfy the first
condition above), and
its own satisfaction is insensitive to the addition of silent events. The following lemma follows from a simple
structural induction.

\begin{lemma}\label{lem:stuttering}
For any trace set $T$ over $\Sigma_\AP$ and closed alternation-free \hmtl{} formula $\phi = \mathcal{Q} \psi$ over $\AP$ ($\mathcal{Q}$ is either a block of existential quantifiers or universal quantifiers and $\psi$ is quantifier-free), $T \models \phi$ iff $\mathit{stutter}(T) \models^\mathit{sync} \mathit{stutter}(\phi)$.
\end{lemma}
The following two lemmas then follow from Lemma~\ref{lem:stuttering}
and the fact that for alternation-free \hmtl{} formulae, satisfiability
in the synchronous semantics can be reduced (in the same way as \hltl{})
to \mtl{} satisfiability.

\begin{lemma}\label{lem:existsdecidable}
The satisfiability problem for $\exists^\ast$-\hmtl{} is decidable.
\end{lemma}
\begin{lemma}
The satisfiability problem for $\forall^\ast$-\hmtl{} is decidable.
\end{lemma}
Lemma~\ref{lem:stuttering},  however,  does not extend to larger
fragments of \hmtl{}. For example, consider $T = \{ (\{p\}, 1)(\{r\}, 3), (\{q\}, 2) \}$
and $\phi = \exists \pi_a \, \forall \pi_b \, (\weventually p_a \wedge \neg \weventually q_b)$.
Now it is obvious that $T \centernot \models \phi$, but since $(\{p\}, 1)(\{r\}, 3) \in \mathit{stutter}(T)$, 
we have $\mathit{stutter}(T) \models^\mathit{sync} \mathit{stutter}(\phi)$ (provided that 
the definition of $\mathit{stutter}(\cdot)$ is extended to general \hmtl{} formulae, as in Lemma~\ref{lem:syncinasync}).
Still, it is not hard to see that the crucial observation used in 
$\exists^\ast \forall^\ast$-\hltl{} satisfiability 
(if $\exists \pi_0 \, \dots \exists \pi_k \, \forall \pi_0' \, \dots \forall \pi_\ell' \, \psi$ is satisfiable, then it is also satisfiable 
by the trace set $\{ \pi_0, \dots \pi_k \}$)
extends to \hmtl{} in the asynchronous semantics; the following lemma
then follows from  Lemma~\ref{lem:existsdecidable}.
\begin{lemma}
The satisfiability problem for $\exists^\ast \forall^\ast$-\hmtl{} is decidable.
\end{lemma}
Finally, note that the undecidability of $\forall \exists$-\hltl{}
carries over to \hmtl{}: in the synchronous semantics, the reduction in~\cite{FinkbeinerH16}
applies directly with some trivial modifications (as we work with finite traces);
undecidability then holds for the case of asynchronous semantics as well, by Lemma~\ref{lem:syncinasync}.
\begin{lemma}
The satisfiability problem for $\forall \exists$-\hmtl{} is undecidable.
\end{lemma}

\begin{theorem}
The satisfiability problem for \hmtl{} is decidable if
the formula does not contain $\forall \exists$.
\end{theorem}

\section{Model checking}

\paragraph{The alternation-free case.}
Without loss of generality, we consider only the case of $\exists^\ast$-\hmtl{}
in the asynchronous semantics.
By Lemma~\ref{lem:syncinasync}, checking $\sem{\A} \models \phi$ (for a \ta{} $\A$ over $\Sigma_\AP$
and a closed $\exists^\ast$-\hmtl{} formula $\phi$ over $\AP$)
is equivalent to checking $\mathit{stutter}(\sem{\A}) \models^\mathit{sync} \mathit{stutter}(\phi)$.
To this end, we define $\mathit{stutter}(\A)$ as the \ta{} over $\Sigma_{\AP_\epsilon}$
obtained from $\A$ by adding a self-loop
labelled with $\{ p^\epsilon \}$ to each location; it should be clear that $\sem{\mathit{stutter}(\A)} = \mathit{stutter}(\sem{\A})$. In this way, the problem reduces to model checking $\exists^\ast$-\hmtl{}
in the synchronous semantics which, as the model-checking problem for $\exists^\ast$-\hltl{},
can be reduced to \mtl{} model checking.
\begin{theorem}
Model checking alternation-free \hmtl{} is decidable.
\end{theorem}

\paragraph{The general case.}
Recall that the model-checking problem for \hltl{} is decidable
even when the specification involves arbitrary nesting of quantifiers.
This is unfortunately not the case for \hmtl{}: allowing only one quantifier
alternation already leads to undecidability.
To see this, recall that any \ta{} can be
written as a formula $\exists X \, \psi$ where $X$
is a set of (new) atomic propositions and $\psi$ is an \mtl{} formula~\cite{HenRas98, Raskin1999}. The undecidable \ta{} universality problem---given a \ta{} $\A$ over $\Sigma$, deciding whether $\sem{\A} = T\Sigma^\ast$---can thus be reduced
to model checking \hmtl{}:
one simply checks whether there exists an $X$-labelling
for every timed word over $\Sigma$ so that $\psi$ is satisfied.
Here we show that model checking \hmtl{} is, essentially, a harder problem;
in the case of asynchronous semantics, it remains undecidable 
even when \emph{both the model and the specification are deterministic and
only one of them uses a single clock}
(i.e.~the other is untimed); by contrast,
\ta{} universality (over finite timed words) is decidable when the \ta{} 
is deterministic~\cite{AluDil94} \emph{or} uses only one clock~\cite{OuaknineW04}.

We adapt the undecidability proof of the
\emph{reactive synthesis} problem for \mtl{} in~\cite{BrihayeEGHMS16}, which itself is by reduction
from the halting problem for \emph{deterministic channel machines} (DCMs), known to be undecidable~\cite{Brand1983}.
Note that the reactive synthesis problem is decidable
when the specification is deterministic~\cite{DSouzaMadhusudan02}
(as a matter of fact, the quantifier-free part $\psi$ in the encoding mentioned above is already in
\ltlhistory{}); 
in this sense, quantification over traces is more powerful than
quantification over strategies (for all possible strategies
of the environment, there is a strategy of the controller that is winning).
For our purpose, we introduce the $\history_I$ operator, in which we allow $I$ to be singular
(note that this is merely syntactic sugar and does not increase the expressiveness of \mtl{}~\cite{HenRas98, Raskin1999}):
\begin{itemize}
\item $(\rho, t)\models \history_I \phi$ iff there exists
  $t'$, $0 \leq t' < t$ such that $t-t' \in I$, $(\rho, t') \models \top$, $(\rho, t') \models \phi$, and $(\rho, t'') \centernot \models \phi$ for all $t''$ such that
$t'' \in (t', t)$ and $(\rho, t'') \models \top$.
\end{itemize}
Let \ltlhistory{} be the fragment of \mtl{} where
all timed subformulae must be of the form $\history_I \phi$,
and all $\phi$'s in such subformulae must be `pure past' formulae;
these requirements ensure that \ltlhistory{}, in which we will write
the quantifier-free part of the specification, translates into
deterministic \ta{s}~\cite{DoyenGRR09}.
To ease the understanding, we will first do the proof for the case of asynchronous semantics
and then adapt it to the case of synchronous semantics.

\begin{theorem}\label{thm:efundecasync}
Model checking $\exists^\ast \forall^\ast$-\hmtl{} and $\forall^\ast \exists^\ast$-\hmtl{} are undecidable in the asynchronous semantics. 
\end{theorem}
\begin{proof}
A DCM $\mathcal{S} = \langle S, s_0, s_\mathit{halt}, M, \Delta \rangle$
can be seen as a finite automaton equipped with an unbounded fifo channel:
$S$ is a finite set of locations, $s_0$ is the initial location, $s_\mathit{halt}$
is the halting location (such that $s_\mathit{halt} \neq s_0$), $M$ is a finite set of messages, and
$\Delta \subseteq S \times \{ m!, m? \mid m \in M \} \times S$
is the transition relation satisfying the following determinism hypothesis: 
(i) $(s, a, s') \in \Delta$ and $(s, a, s'') \in \Delta$ implies $s' = s''$;
(ii) if $(s, m!, s') \in \Delta$ then it is the only outgoing transition
from $s$. 
Without loss of generality, we further assume that 
there is no incoming transition to $s_0$, no outgoing transition from $s_\mathit{halt}$,
and $(s_0, a, s') \in \Delta$ implies that $a \in \{ m! \mid m \in M \}$ and $s' \neq s_\mathit{halt}$.
The semantics of $\mathcal{S}$ can be described with a graph $G(\mathcal{S})$ with vertices 
$\{ (s, x) \mid s \in S\text{, } x \in M^\ast \}$
and edges defined as follows: (i) $(s, x) \rightarrow (s', xm)$
if $(s, m!, s') \in \Delta$; (ii) $(s, mx) \rightarrow (s', x)$
if $(s, m?, s') \in \Delta$.
In other words, $m!$ `writes' a copy of $m$ to the channel
and $m?$ `reads' a copy of $m$ off the channel.
We say that $\mathcal{S}$ \emph{halts}
if there is a path in $G(\mathcal{S})$ from $(s_0, \epsilon)$
to $(s_\mathit{halt}, x)$ (a \emph{halting computation} of $\mathcal{S}$) for some $x \in M^\ast$.

The idea, as in many similar proofs (e.g.,~\cite{OuaWor07}), is to encode a halting computation of $\mathcal{S}$ as a
trace where each $m?$ is preceded
by a corresponding $m!$ exactly $1$ time unit earlier,
and each $m!$ is followed by an $m?$ exactly $1$ time unit later
if $s_\mathit{halt}$ has not been reached yet.
To this end, let the model $\A$ be an (untimed) finite automaton 
over $\Sigma = 2^\AP$ where $\AP = \{ m^!, m^? \mid m \in M \} \cup \{ p^\mathit{begin}, p^\mathit{end}, p^\mathit{read}, 
p^1, q^1 \}$
and whose set of locations is $S \cup \{s_1\}$,
where $s_1$ is a new non-accepting location.
The transitions of $\A$ follow $\mathcal{S}$:
for each $m \in M$, $s \xrightarrow{\{m^?\}} s'$ is a transition
of $\A$ iff $(s, m?, s') \in \Delta$, and similarly for $m!$---except for those going out of $s_0$
or going into $s_\mathit{halt}$, on which we further require
$p^\mathit{begin}$ or $p^\mathit{end}$ to hold, respectively.
Let $s_0$ be the initial location and $s_\mathit{halt}$ be the only
accepting location, and finally add transitions
$s_0 \xrightarrow{\{p^\mathit{read}\}} s_\mathit{halt}$
and $s_0 \xrightarrow{\{p^1\}} s_1 \xrightarrow{\{q^1\}} s_\mathit{halt}$.
It is clear that $\A$ is deterministic and it admits only three types of traces: 
\begin{enumerate}
\item From $s_0$ through some other locations of $\mathcal{S}$ and finally $s_\mathit{halt}$, i.e.~those respecting the transition
relation, but not necessarily the semantics, of $\mathcal{S}$.
\item From $s_0$ to $s_\mathit{halt}$ in a single transition (on which $p^\mathit{read}$ holds).
\item From $s_0$ to $s_1$ and then $s_\mathit{halt}$.
\end{enumerate}
It remains to write a specification $\phi$ such that $\sem{\A} \models \phi$ exactly when
$\A$ accepts a trace of type ($1$) that also respects the semantics of $\mathcal{S}$.
This is where the traces of types ($2$) and ($3$) come into play: for example, if a
trace of type ($1$) issues a read $m?$ 
without a corresponding write $m!$, then a trace of 
type ($3$) can be used to `pinpoint' the error. 
More precisely, let $\phi = \exists \pi_a \, \forall \pi_b \, (\psi_1 \wedge \psi_2 \wedge \psi_3 \wedge \psi_4)$ where
\begin{itemize}
\item $\psi_1 = \weventually p^\mathit{end}_{a}$ ensures that $\pi_a$ is of type ($1$); 
\item $\psi_2 = \weventually (p^\mathit{read}_b \wedge \psi_\mathit{R}) \Rightarrow \weventually (p^\mathit{read}_b \wedge \history_{\geq 1} p^\mathit{begin}_a)$, where
$\psi_\mathit{R} = \bigvee \{ m^?_a \mid m \in M \}$, is a simple sanity check
which ensures that in $\pi_a$, each $m^?$ must happen at time $\geq t + 1$ if $p^\mathit{begin}$ happens at $t$;
\item $\psi_3 = \bigwedge_{m \in M} \Big( \weventually (q^1_b \wedge m_a^?) \Rightarrow \big(\weventually (p^\mathit{begin}_a \wedge \weventually p^1_b) \wedge \weventually (q^1_b \wedge \history_{=1} p^1_b) \Rightarrow \weventually (p^1_b \wedge m_a^!) \big) \Big)$ ensures that each $m^?$, if it happens at $t$, is preceded by
a corresponding $m^!$ at $t - 1$ in $\pi_a$;
\item $\psi_4 = \bigwedge_{m \in M} \Big( \weventually (p^1_b \wedge m_a^!) \Rightarrow \weventually (p^\mathit{end}_a \wedge \history_{<1} p^1_b) \vee \big(\weventually (q^1_b \wedge \history_{=1} p^1_b) \Rightarrow \weventually (q^1_b \wedge m_a^?) \big) \Big)$ ensures that each $m^!$ at $t$ is followed by a corresponding $m^?$ at $t+1$ (unless $p^\mathit{end}$
happens first) in $\pi_a$.
\end{itemize}
Now observe that the only timed subformulae are $\history_{\geq 1} p^\mathit{begin}_a$, $\history_{=1} p^1_b$, 
and $\history_{<1} p^1_b$.
As $p^1$ and $p^\mathit{read}$ cannot happen in the same trace ($\pi_b$), 
it is not hard to see that the reduction remains correct if we replace these by 
$\history_{\geq 1} (p^\mathit{begin}_a \vee p^1_b)$, $\history_{=1} (p^\mathit{begin}_a \vee p^1_b)$, 
and $\history_{<1} (p^\mathit{begin}_a \vee p^1_b)$ (respectively) to obtain $\psi_2'$, $\psi_3'$, and $\psi_4'$.
It follows that $\psi_1 \wedge \psi_2' \wedge \psi_3' \wedge \psi_4'$
can be translated into a one-clock deterministic \ta{}.
Finally, it is possible to move all the timing constraint
into the model and use an untimed \ltl{} formula as the specification:
in the model, ensure that $p^1$ and $q^1$ are separated by exactly $1$ time unit,
and add $s_0 \xrightarrow{\{p^2\}} s_1 \xrightarrow{\{q^2\}} s_\mathit{halt}$ such that
$p^2$ and $q^2$ are separated by $<1$ time unit;
in the specification, use $p^2$, $q^2$ to rule out those $\pi_a$'s
with some $m^?$ at $<1$ time unit from $p^\mathit{begin}$.
\end{proof}

Now we consider the synchronous semantics. The corresponding result is weaker in this case, as we will see in the next section that in several subcases the problem becomes decidable. Still, the reduction above can be made to work if the model has one clock and an extra trace quantifier is allowed.
\begin{theorem}\label{thm:efundecsync}
Model checking $\exists^\ast \forall^\ast$-\hmtl{} and $\forall^\ast \exists^\ast$-\hmtl{} are undecidable in the synchronous semantics. 
\end{theorem}
\begin{proof}
We use a modified model $\A'$ whose set of locations is $S \cup \{ s_1, s_2, s_3, s_4 \}$;
the transitions are similar to $\A$ in the proof of Theorem~\ref{thm:efundecasync}, but we now use a clock $x$ in the path
$s_0 \xrightarrow[x := 0]{\{p^1\}} s_1 \xrightarrow[x \geq 1, x:=0]{\{q^1\}} s_\mathit{halt}$, the paths $s_0 \xrightarrow[x := 0]{\{p^2\}} s_2 \xrightarrow[x \leq 1, x:=0]{\{q^2\}} s_\mathit{halt}$, $s_0 \xrightarrow[x := 0]{\{p^3\}} s_3 \xrightarrow[x > 1, x:=0]{\{q^3\}} s_\mathit{halt}$, and $s_0 \xrightarrow[x := 0]{\{p^4\}} s_4 \xrightarrow[x < 1, x:=0]{\{q^4\}} s_\mathit{halt}$ are added, and
$s_0 \xrightarrow{\{p^\mathit{read}\}} s_\mathit{halt}$ is removed.
Moreover, a self-loop labelled with $\{p^\epsilon\}$
is added to each of $s_0$, $s_1$, $s_2$, $s_3$, $s_4$, and $s_\mathit{halt}$. 
The specification is $\phi' = \exists \pi_a \, \forall \pi_b \, \forall \pi_c \bigwedge_{1 \leq i \leq 9} \psi_i'$ where
$\bigwedge_{1 \leq i \leq 9} \psi_i'$ is the following untimed \ltl{} formula:
\begin{itemize}
\item $\psi_1' = \weventually p^\mathit{end}_{a}$; 
\item $\psi_2' = \weventually (q^4_b \wedge \psi_\mathit{R}) \Rightarrow \neg \weventually (p^4_b \wedge p^\mathit{begin}_a)$ where
$\psi_\mathit{R} = \bigvee \{ m^?_a \mid m \in M \}$;
\item $\psi_3' = \bigwedge_{m \in M} \big( \weventually (q^1_b \wedge q^2_c \wedge m_a^?) \wedge \weventually (p^1_b \wedge p^2_c) \Rightarrow \weventually (p^1_b \wedge p^2_c \wedge m_a^!) \big)$;
\item $\psi_4' = \weventually (q^3_b \wedge \psi_\mathit{R}) \Rightarrow \neg \weventually (p^3_b \wedge \nextx q^3_b)$;
\item $\psi_5' = \weventually (q^3_b \wedge q^4_c \wedge \psi_\mathit{R}) \Rightarrow \neg \weventually (p^3_b \wedge \nextx p^4_c)$;
\item $\psi_6' = \weventually (p^4_b \wedge \psi_\mathit{W}) \Rightarrow \neg \weventually (q^4_b \wedge \neg \nextx \top)$ where
$\psi_\mathit{W} = \bigvee \{ m^!_a \mid m \in M \}$;
\item $\psi_7' = \bigwedge_{m \in M} \Big( \weventually (p^1_b \wedge p^2_c \wedge m_a^!) \wedge \weventually (q^1_b \wedge q^2_c) \Rightarrow \weventually \big(q^1_b \wedge q^2_c \wedge (m_a^? \vee p_a^\epsilon) \big) \Big)$;
\item $\psi_8' = \weventually (p^3_b \wedge \psi_\mathit{W}) \Rightarrow \neg \weventually (p^3_b \wedge \nextx q^3_b)$;
\item $\psi_9' = \weventually (p^3_b \wedge p^4_c \wedge \psi_\mathit{W}) \Rightarrow \neg \weventually (q^3_b \wedge \nextx q^4_c)$.
\end{itemize}
In this modified reduction, $\psi_1'$, $\psi_2'$ play similar roles as $\psi_1$, $\psi_2$ in the proof of Theorem~\ref{thm:efundecasync}. $\psi_3'$ ensures that if each $m^?$ at $t$ is preceded by an event at $t-1$, then $m^!$ must hold there.
$\psi_4'$ and $\psi_5'$ ensures that each $m^?$ at $t$ is actually preceded by an event at $t-1$.
The roles of $\psi_6'$, $\psi_7'$, $\psi_8'$, and $\psi_8'$ are analogous (note the use of silent events at
the end of $\pi_a$).
\end{proof}

\section{Decidable subcases}
While the negative results in the previous section may be disappointing,
we stress again that model checking alternation-free \hmtl{} 
is no harder than \mtl{} model checking, and it can in fact be carried
out with algorithms and tools for the latter. 
In any case, we now identify several subcases where model checking
is decidable beyond the alternation-free fragment.
\paragraph{Untimed model + untimed specification.}
The first case we consider is when both the model and the specification 
are untimed, and the asynchronous semantics is assumed (this case is simply \hltl{} model checking
in the synchronous semantics).
Our algorithm follows the lines of~\cite{ClarksonFKMRS14} and 
is essentially based on \emph{self-composition} (cf.~\cite{BartheDR11}, and many others; see the references in~\cite{ClarksonFKMRS14}) of the model;
the difficulty here, however, is to handle interleaving of events.
Let the model $\A$ be a finite automaton over $\Sigma_\AP$ and the specification be
a (untimed) closed \hmtl{} formula over $\AP$.
Without loss of generality, we assume the specification to be
$\phi = \exists \pi_1 \, \forall \pi_2 \, \dots \exists \pi_{k-1} \, \forall \pi_k \, \psi$,
which can be rewritten into $\exists \pi_1 \, \neg \exists \pi_2 \, \neg \dots \neg \exists \pi_{k-1} \, \neg \exists \pi_k \, \neg \psi$.
We start by translating $\mathit{stutter}(\neg \psi)$ (in which we
replace all occurrences of $\top_i$ with $\neg p^\epsilon_i$, i.e.~regarded here simply as
an \mtl{} formula over $(\AP_\epsilon)^k = \{ p_i \mid p \in \AP_\epsilon, 1 \leq i \leq k \}$) into the equivalent finite automaton
over $\Sigma_{(\AP_\epsilon)^k}$, and take its product with
(i) the automaton for $\wglobally(\bigvee_{1 \leq i \leq k} \neg p^\epsilon_i ) \wedge \big(\bigwedge_{1 \leq i \leq k} \wglobally(p^\epsilon_i \Rightarrow \bigwedge_{p \in \AP} \neg p_i)\big)$ and (ii) the automaton obtained from $\mathit{stutter}(\A)$ by extending the alphabet
to $\Sigma_{(\AP_\epsilon)^k}$ and renaming all the occurrences of $p$ to $p_k$, to obtain $\B$.
Now let $\mathcal{C}$ be the projection of $\B$ onto $(\AP_\epsilon)^{k-1} = \{ p_i \mid p \in \AP_\epsilon, 1 \leq i \leq k - 1 \}$ (this step corresponds to $\exists$ in $\neg \exists \pi_k$). By construction, $\B$ accepts only traces that are well-formed in
dimensions $1$ to $k - 1$, and so does $\mathcal{C}$; but $\mathcal{C}$ may accept traces containing $\{ p^\epsilon_i \mid 1 \leq i \leq k - 1\}$-events. 
We replace these events by $\epsilon$ (the `real' silent event, which can be removed with the standard textbook constructions, e.g.,~\cite{HopcroftUllman}) to obtain $\mathcal{C}'$.
Finally, we complement $\mathcal{C}'$ to obtain $\mathcal{C}''$ (this step corresponds to $\neg$ in $\neg \exists \pi_k$). 
We can then start over by taking the product of $\mathcal{C}''$,
the automaton for $\wglobally(\bigvee_{1 \leq i \leq k - 1} \neg p^\epsilon_i ) \wedge \big(\bigwedge_{1 \leq i \leq k - 1} \wglobally(p^\epsilon_i \Rightarrow \bigwedge_{p \in \AP} \neg p_i)\big)$, and the automaton obtained from
$\mathit{stutter}(\A)$ by extending the alphabet to $\Sigma_{(\AP_\epsilon)^{k-1}}$ and renaming all the occurrences
of $p$ to $p_{k-1}$; the resulting automaton is the new $\B$. We continue this process until
the outermost quantifier $\exists \pi_1$ is reached, when we test the emptiness of $\B$ (at this point, it is an automaton
over $\Sigma_{\AP_\epsilon}$).

\begin{proposition}\label{prop:purelyuntimed}
Model checking \hmtl{} is decidable when the model and the specification are both untimed.
\end{proposition}

\paragraph{One clock + one alternation.}
The algorithm outlined in the previous case crucially depends on the fact that both $\A$ and $\phi$ are untimed,
hence their product (in the sense detailed in the previous case) can be complemented.
When the synchronous semantics is assumed and there are is only one quantifier alternation in $\phi$,
it might be the case that we do not actually need complementation.
For example, if $\A$ is untimed and $\phi = \forall \pi_a \, \exists \pi_b \, \exists \pi_c \, \psi$
where $\psi$ translates into a one-clock \ta{}, the corresponding model-checking problem
clearly reduces to universality for one-clock \ta{s}, which is decidable but non-primitive
recursive~\cite{AbdullaDOQW08}.\footnote{This case is undecidable in the asynchronous semantics
by Theorem~\ref{thm:efundecasync}; as explained above, the algorithm may introduce $\epsilon$-transitions
in the asynchronous semantics, while universality for one-clock \ta{s} with $\epsilon$-transitions is undecidable~\cite{AbdullaDOQW08}.}
This observation applies to other cases as well, such as when $\A$ is a one-clock \ta{}
and $\phi = \exists \pi_a \, \forall \pi_b \, \psi$ where $\psi$ is untimed;
here model checking reduces to language inclusion between two one-clock \ta{s}.

\paragraph{Untimed model + \mia{} specification.}
The main obstacle in applying the algorithm above to larger fragments of
\hmtl{}, as should be clear now, is that universal quantifiers amount to complementations,
which are not possible in general in the case of \ta{s}.
Moreover, we note that the usual strategy of restricting to deterministic models and specifications
does not help, as the projection step in the algorithm
necessarily introduces non-determinism.
To make the algorithm work for larger fragments, we essentially need a class of automata that is both
\emph{closed under projection} and \emph{complementable}.
Fortunately, there is a subclass of one-clock \ta{s} that satisfies these conditions.
We consider two additional restrictions on one-clock \ta{s}: 
\begin{itemize}
\item Non-Singular (NS): a one-clock \ta{} is NS if all the guards are non-singular.
\item Reset-on-Testing (RoT): a one-clock \ta{} is RoT if whenever the guard
of a transition is not $\top$, the single clock $x$ must be reset on that transition.
\end{itemize}
One-clock \ta{s} satisfying both NS and RoT are called \emph{metric interval automata} (\mia{s}), which 
are determinisable~\cite{Ferrere18}.
Since the projection operation cannot invalidate NS and RoT,
the algorithm above can be applied when the synchronous semantics is assumed,
$\A$ is untimed, $\psi$ or $\neg \psi$ translates to a \mia{}, and only one complementation is involved; in this case it runs in elementary time.
\begin{proposition}\label{prop:mia}
Model checking $\forall^\ast \exists^\ast$-\hmtl{} ($\exists^\ast \forall^\ast$-\hmtl{}) is decidable in the synchronous semantics when the model is untimed and
$\psi$ ($\neg \psi$) translates into a \mia{} in the specification $\phi = \forall \pi_1 \, \dots \exists \pi_k \, \psi$ ($\phi = \exists \pi_1 \, \dots \forall \pi_k \, \psi$). 
\end{proposition}
On the other hand, we can adapt the proof of Theorem~\ref{thm:efundecsync}
to show that model checking an untimed model against
an $\exists^\ast \forall^\ast$-\hmtl{} specification $\phi$ in the synchronous semantics, when the quantifier-free part $\psi$ (instead of $\neg \psi$)
translates into a \mia{}, remains undecidable.

\begin{proposition}\label{prop:effundecns}
Model checking $\exists^\ast \forall^\ast$-\hmtl{} is undecidable in the synchronous semantics when the model is untimed and $\psi$ in the specification $\phi = \exists \pi_1 \, \dots \forall \pi_k \, \psi$ translates into a \mia{}. 
\end{proposition}
\begin{proof}[sketch] Similar to the proof of Theorem~\ref{thm:efundecsync}:
we add a new initial location $s_0'$ and the transitions $s_0' \xrightarrow{\{p^\epsilon\}} s_0$, 
self-loops labelled with $\{p^\epsilon\}$ to $s_0$ and $s_\mathit{halt}$,
and $s_0 \xrightarrow{\{p^3\}} s_3 \xrightarrow{\{q^3\}} s_\mathit{halt}$, $s_0 \xrightarrow{\{p^4\}} s_4 \xrightarrow{\{q^4\}} s_\mathit{halt}$, $s_0 \xrightarrow{\{q^5\}} s_\mathit{halt}$.
In the quantifier-free part $\psi''$ of the specification $\phi'' = \exists \pi_a \, \forall \pi_b \, \forall \pi_c \, \psi''$ we use, e.g.,~$\bigwedge_{m \in M} \Big( \big(\weventually (q^3_b \wedge q^4_c \wedge m_a^?) \wedge \weventually (q^3_b \wedge \history_{>1} p^3_b) \wedge \weventually (q^4_c \wedge \history_{<1} p^4_c) \big) \vee \big( \weventually (q^3_b \wedge q^5_c \wedge m_a^?) \wedge \weventually (q^3_b \wedge \history_{>1} p^3_b) \big) \Rightarrow \weventually \big(p^3_b \wedge \eventually (m_a^! \wedge \eventually p^4_c )\big) \vee \big(p^3_b \wedge \eventually (m_a^! \wedge \eventually q^5_c )\big) \Big)$. Finally, note that $\psi''$ can be translated into a \mia{}.
\end{proof}
The decidability results in the synchronous semantics are summarised in Table~\ref{tab:mc}.

\begin{table}[h]
\centering
\begin{tabular}{ccccc}
\toprule
\diagbox{Model}{Spec.} & untimed & NS$+$RoT & NS & RoT \\
\otoprule
untimed & \makecell{Dec.\\(Proposition~\ref{prop:purelyuntimed})} & \makecell{Dec. for $\forall^\ast \exists^\ast$\\(Proposition~\ref{prop:mia})} & \makecell{Undec. for $\exists \forall \forall$ \\(Proposition~\ref{prop:effundecns})} & \makecell{Undec. for $\exists \forall \forall$ \\(Proposition~\ref{prop:effundecns})}  \\
\midrule
NS+RoT & \makecell{Undec. for $\exists \forall \forall$ \\(Theorem~\ref{thm:efundecsync})} & Undec. & Undec. & Undec. \\
\midrule
NS & Undec. & Undec. & Undec. & Undec. \\
\midrule
RoT & Undec. & Undec. & Undec. & Undec. \\
\bottomrule
\end{tabular}
\caption{Decidability of model checking untimed or one-clock \ta{s} against (one-clock) \hmtl{}
in the synchronous semantics; NS stands for Non-Singular constraints and RoT stands for Reset-on-Testing.}
\label{tab:mc}
\end{table}

%

\paragraph{Bounded time domains.}

We end this section by showing that when there is an \emph{a priori} bound $N$ (where $N$ is a positive integer) on the
length of the time domain, the model-checking problem for full \hmtl{} becomes decidable;
in fact, in the case of synchronous semantics it reduces to the satisfiability problem for \qptl{}~\cite{Sistla1985}.
From a practical point of view, this implies that \emph{time-bounded} \hmtl{}
verification (at least for the $\exists^\ast \forall^\ast$-fragment, say) can be carried out with highly efficient, off-the-shelf tools that works
with \ltl{} and (untimed) automata, such as SPOT~\cite{Duret-Lutz2016}, GOAL~\cite{TsaiTH13}, and \texttt{Owl}~\cite{KretinskyMS18}.

We assume the asynchronous semantics.
For a given $N$,
we consider all traces in which all timestamps are less than $N$.
Denote by $\sem{\A}_{[0, N)}$ the set of all such traces in $\sem{\A}$;
the model-checking problem then becomes deciding whether $\sem{\A}_{[0, N)} \models \phi$.
As before, we assume $\phi$ to be
$\exists \pi_1 \, \neg \exists \pi_2 \, \neg \dots \neg \exists \pi_{k-1} \, \neg \exists \pi_k \, \neg \psi$.
Following~\cite{Ouaknine2009, Ho14}, we can 
use the \emph{stacking construction} to obtain, from the conjunction $\psi'$ of $\mathit{stutter}(\neg \psi)$
and $\wglobally(\bigvee_{\pi \in \mathcal{Q}} \neg p^\epsilon_\pi ) \wedge \big(\bigwedge_{\pi \in \mathcal{Q}} \wglobally(p^\epsilon_\pi \Rightarrow \bigwedge_{p \in \AP} \neg p_\pi)\big)$,
an equi-satisfiable untimed (\qptl{}) formula $\overline{\phi} = \exists W \, \overline{\psi'}$ over the stacked alphabet $\overline{(\AP_\epsilon)^k} \cup \overline{Q}$
(where $\overline{(\AP_\epsilon)^k} = \{ p_{i, j} \mid p \in \AP_\epsilon\text{, } 1 \leq i \leq k \text{, } 0 \leq j < N \}$ and $\overline{Q} = \{ q_j \mid 0 \leq j < N \}$). 
We apply the following modifications to $\overline{\phi}$ to obtain $\overline{\phi'}$:
\begin{itemize}
\item introduce atomic propositions $\{ p^\epsilon_i  \mid 1 \leq i \leq k \}$ and add the formula $\bigwedge_{1 \leq i \leq k} \wglobally \Big( \big(\bigwedge_{0 \leq j < N} (q_j \Rightarrow p^\epsilon_{i, j}) \big) \Leftrightarrow p^\epsilon_i \Big)$ as a conjunct;
\item introduce atomic propositions $\{ q_{i, j} \mid 1 \leq i \leq k \text{, } 0 \leq j < N \}$ and add the formula
$\bigwedge_{1 \leq i \leq k} \wglobally \big( \bigwedge_{0 \leq j < N} ( \neg p^\epsilon_i \wedge q_j \Leftrightarrow q_{i, j} ) \big)$ as a conjunct;
\item project away $\{p^\epsilon_{i, j} \mid 1 \leq i \leq k \text{, } 0 \leq j < N \}$ and $\overline{Q}$.
\item replace all occurrences of $p^\epsilon_i$ by $\bot_i$.
\end{itemize}
Now, as we mentioned earlier, we can write $\A$ as an (\msoone~\cite{Ouaknine2009}) formula
$\phi_\A = \exists X_\A \, \psi_\A$ where $X_\A$ is a set of atomic propositions such that $\AP \cap X_\A = \emptyset$
and $\psi_\A$ is an \mtl{} formula over $\AP \cup X_\A$.
Let $\overline{\phi_\A}$ be its stacked counterpart $\exists \overline{X_\A} \, \exists Y \, \overline{\psi_\A}$;
we translate $\overline{\phi_\A}$ back into an untimed automaton $\overline{\A}$ over the stacked alphabet $\overline{\AP} \cup \overline{Q}$.
The problem thus reduces to untimed model checking of $\overline{\A}$ against $\exists \pi_1 \, \forall \pi_2 \, \dots \exists \pi_{k-1} \, \forall \pi_k \, \overline{\phi'}$ in the asynchronous semantics, which is decidable by Proposition~\ref{prop:purelyuntimed} ($\overline{\phi'}$ has outermost existential propositional quantifiers, but clearly the equivalent automaton can be used 
directly in the algorithm).

Finally, note that the proof is simpler for the case of synchronous semantics: we can simply
work with a (non-stuttering) \msoone{} formula in all the intermediate steps without translating it into an automaton,
and then check the satisfiability of the final formula by stacking it into a \qptl{} formula.

\begin{proposition}
Model checking \hmtl{} is decidable
when the time domain is $[0, N)$, where $N$ is a given positive integer.
\end{proposition}

\section{Conclusion}

We studied the satisfiability and model-checking problems
for \hmtl{} over sets of timed words.
While satisfiability can be solved
similarly as for \hltl{}, model checking turned out to
be undecidable when the specification involves at least one quantifier alternation; this holds even for otherwise fairly restricted models and specifications.
On the other hand, we showed that model checking beyond the alternation-free fragment is possible if (i) interleaving of events in different traces is disallowed, or (ii) the time domain is $N$-bounded for a fixed positive integer $N$.
We leave as future work to investigate whether
a suitable notion of `timing fuzziness' (e.g.,~\cite{AlurTM05, GuptaHJ97, DonzeM10}) can be incorporated, either to recover decidability of model checking or
better align with practical applications, e.g.,~monitoring of cyber-physical systems~\cite{BonakdarpourDP18, BartocciDDFMNS18}.
Another possible direction is to consider the case
where the number of events~\cite{LorberRNA17} (or more generally, the number
of events in any interval of fixed length~\cite{Wilke1994, FuriaR08}) in any trace is bounded, which may be sufficient for
modelling many real-world systems. 

\bibliographystyle{splncs04}
\bibliography{refs}

\begin{thebibliography}{10}
\providecommand{\url}[1]{\texttt{#1}}
\providecommand{\urlprefix}{URL }
\providecommand{\doi}[1]{https://doi.org/#1}

\bibitem{AbdullaDOQW08}
Abdulla, P.A., Deneux, J., Ouaknine, J., Quaas, K., Worrell, J.: Universality
  analysis for one-clock timed automata. Fundam. Inform  \textbf{89}(4),
  419--450 (2008)

\bibitem{AbrahamB18}
{\'A}brah{\'a}m, E., Bonakdarpour, B.: {HyperPCTL}: A temporal logic for
  probabilistic hyperproperties. In: QEST. Lecture Notes in Computer Science,
  vol. 11024, pp. 20--35. Springer (2018)

\bibitem{AgrawalB16}
Agrawal, S., Bonakdarpour, B.: Runtime verification of k-safety hyperproperties
  in {HyperLTL}. In: CSF. pp. 239--252. IEEE Computer Society (2016)

\bibitem{AluDil94}
Alur, R., Dill, D.L.: A theory of timed automata. Theoretical Computer Science
  \textbf{126}(2),  183--235 (1994)

\bibitem{AluFed96}
Alur, R., Feder, T., Henzinger, T.A.: The benefits of relaxing punctuality.
  Journal of the ACM  \textbf{43}(1),  116--146 (1996)

\bibitem{Alur1992}
Alur, R., Henzinger, T.A.: Back to the future: towards a theory of timed
  regular languages. In: FOCS. pp. 177--186. IEEE Computer Society Press (1992)

\bibitem{AluHen93}
Alur, R., Henzinger, T.A.: Real-time logics: Complexity and expressiveness.
  Information and Computation  \textbf{104}(1),  35--77 (1993)

\bibitem{AluHen94}
Alur, R., Henzinger, T.A.: A really temporal logic. Journal of the ACM
  \textbf{41}(1),  164--169 (1994)

\bibitem{AlurTM05}
Alur, R., Torre, S.L., Madhusudan, P.: Perturbed timed automata. In: HSCC.
  LNCS, vol.~3414, pp. 70--85. Springer (2005)

\bibitem{BartheDR11}
Barthe, G., D'Argenio, P.R., Rezk, T.: Secure information flow by
  self-composition. Mathematical Structures in Computer Science
  \textbf{21}(6),  1207--1252 (2011)

\bibitem{BartocciDDFMNS18}
Bartocci, E., Deshmukh, J.V., Donz{\'e}, A., Fainekos, G.E., Maler, O.,
  Nickovic, D., Sankaranarayanan, S.: Specification-based monitoring of
  cyber-physical systems: A survey on theory, tools and applications. In:
  Lectures on Runtime Verification, LNCS, vol. 10457, pp. 135--175. Springer
  (2018)

\bibitem{BonakdarpourDP18}
Bonakdarpour, B., Deshmukh, J.V., Pajic, M.: Opportunities and challenges in
  monitoring cyber-physical systems security. In: ISoLA. LNCS, vol. 11247, pp.
  9--18. Springer (2018)

\bibitem{Brand1983}
Brand, D., Zafiropulo, P.: On communicating finite state machines. Journal of
  the ACM  \textbf{30},  323--342 (1983)

\bibitem{BrihayeEGHMS16}
Brihaye, T., Esti{\'e}venart, M., Geeraerts, G., Ho, H.M., Monmege, B.,
  Sznajder, N.: Real-time synthesis is hard! In: FORMATS. LNCS, vol.~9884, pp.
  105--120. Springer (2016)

\bibitem{Cimatti2002}
Cimatti, A., Clarke, E., Giunchiglia, E., Giunchiglia, F., Pistore, M., Roveri,
  M., Sebastiani, R., Tacchella, A.: {NuSMV2}: An opensource tool for symbolic
  model checking. In: CAV. LNCS, vol.~2404, pp. 359--364. Springer (2002)

\bibitem{ClarksonFKMRS14}
Clarkson, M.R., Finkbeiner, B., Koleini, M., Micinski, K.K., Rabe, M.N.,
  S{\'a}nchez, C.: Temporal logics for hyperproperties. In: POST. LNCS,
  vol.~8414, pp. 265--284. Springer (2014)

\bibitem{ClarksonS10}
Clarkson, M.R., Schneider, F.B.: Hyperproperties. Journal of Computer Security
  \textbf{18}(6),  1157--1210 (2010)

\bibitem{DonzeM10}
Donz{\'e}, A., Maler, O.: Robust satisfaction of temporal logic over
  real-valued signals. In: FORMATS. LNCS, vol.~6246, pp. 92--106. Springer
  (2010)

\bibitem{DoyenGRR09}
Doyen, L., Geeraerts, G., Raskin, J.F., Reichert, J.: Realizability of
  real-time logics. In: FORMATS. LNCS, vol.~5813, pp. 133--148. Springer (2009)

\bibitem{DSouzaMadhusudan02}
D'Souza, D., Madhusudan, P.: Timed control synthesis for external
  specifications. In: STACS, LNCS, vol.~2285, pp. 571--582. Springer (2002)

\bibitem{Duret-Lutz2016}
Duret-Lutz, A., Lewkowicz, A., Fauchille, A., Michaud, T., Renault, E., Xu, L.:
  Spot 2.0 - a framework for {LTL} and $\omega$-automata manipulation. In:
  ATVA. LNCS, vol.~9938, pp. 122--129. Springer (2016)

\bibitem{Ferrere18}
Ferr{\`e}re, T.: The compound interest in relaxing punctuality. In: FM. LNCS,
  vol. 10951, pp. 147--164. Springer (2018)

\bibitem{FinkbeinerH16}
Finkbeiner, B., Hahn, C.: Deciding hyperproperties. In: CONCUR. LIPIcs,
  vol.~59, pp. 13:1--13:14. Schloss Dagstuhl - Leibniz-Zentrum fuer Informatik
  (2016)

\bibitem{FinkbeinerHH18}
Finkbeiner, B., Hahn, C., Hans, T.: Mghyper: Checking satisfiability of
  {HyperLTL} formulas beyond the $\exists^\ast\forall^\ast$ fragment. In: ATVA.
  Lecture Notes in Computer Science, vol. 11138, pp. 521--527. Springer (2018)

\bibitem{FinkbeinerHS17}
Finkbeiner, B., Hahn, C., Stenger, M.: Eahyper: Satisfiability, implication,
  and equivalence checking of hyperproperties. In: CAV. Lecture Notes in
  Computer Science, vol. 10427, pp. 564--570. Springer (2017)

\bibitem{FinkbeinerHST17}
Finkbeiner, B., Hahn, C., Stenger, M., Tentrup, L.: Monitoring hyperproperties.
  In: RV. LNCS, vol. 10548, pp. 190--207. Springer (2017)

\bibitem{FinkbeinerHST18}
Finkbeiner, B., Hahn, C., Stenger, M., Tentrup, L.: {RVHyper}: A runtime
  verification tool for temporal hyperproperties. In: TACAS. Lecture Notes in
  Computer Science, vol. 10806, pp. 194--200. Springer (2018)

\bibitem{FinkbeinerHT18}
Finkbeiner, B., Hahn, C., Torfah, H.: Model checking quantitative
  hyperproperties. In: CAV. Lecture Notes in Computer Science, vol. 10981, pp.
  144--163. Springer (2018)

\bibitem{FinkbeinerRS15}
Finkbeiner, B., Rabe, M.N., S{\'a}nchez, C.: Algorithms for model checking
  {HyperLTL} and {$HyperCTL^\ast$}. In: CAV. LNCS, vol.~9206, pp. 30--48.
  Springer (2015)

\bibitem{FuriaR08}
Furia, C.A., Rossi, M.: Mtl with bounded variability: Decidability and
  complexity. In: FORMATS. LNCS, vol.~5215, pp. 109--123. Springer (2008)

\bibitem{GerkingSB18}
Gerking, C., Schubert, D., Bodden, E.: Model checking the information flow
  security of real-time systems. In: ESSoS. LNCS, vol. 10953, pp. 27--43.
  Springer (2018)

\bibitem{GogMes82}
Goguen, J.A., Meseguer, J.: Security policies and security models. In: S\&P.
  pp. 11--20. {IEEE} Computer Society Press, Oakland, CA (1982)

\bibitem{GuptaHJ97}
Gupta, V., Henzinger, T.A., Jagadeesan, R.: Robust timed automata. In: HART.
  LNCS, vol.~1201, pp. 331--345. Springer (1997)

\bibitem{Heinen2018}
Heinen, J.: Model Checking Timed Hyperproperties. Master's thesis, Saarland
  University (2018)

\bibitem{HenRas98}
Henzinger, T.A., Raskin, J.F., Schobbens, P.Y.: The regular real-time
  languages. In: ICALP. LNCS, vol.~1443, pp. 580--591. Springer (1998)

\bibitem{Ho14}
Ho, H.M.: On the expressiveness of metric temporal logic over bounded timed
  words. In: RP. Lecture Notes in Computer Science, vol.~8762, pp. 138--150.
  Springer (2014)

\bibitem{Holzmann1997}
Holzmann, G.J.: The model checker {SPIN}. IEEE Transactions on Software
  Engineering  \textbf{23}(5),  279--295 (1997)

\bibitem{HopcroftUllman}
Hopcroft, J.E., Ullman, J.D.: Introduction to Automata Theory, Languages, and
  Computation. Addison-Wesley (1979)

\bibitem{HuismanWS06}
Huisman, M., Worah, P., Sunesen, K.: A temporal logic characterisation of
  observational determinism. In: CSFW. p.~3. IEEE Computer Society (2006)

\bibitem{Kocher2018}
Kocher, P., Genkin, D., Gruss, D., Haas, W., Hamburg, M., Lipp, M., Mangard,
  S., Prescher, T., Schwarz, M., Yarom, Y.: Spectre attacks: Exploiting
  speculative execution. CoRR  \textbf{abs/1801.01203} (2018)

\bibitem{Koy90}
Koymans, R.: Specifying real-time properties with metric temporal logic.
  Real-Time Systems  \textbf{2}(4),  255--299 (1990)

\bibitem{KretinskyMS18}
Kret{\'i}nsk{\'y}, J., Meggendorfer, T., Sickert, S.: Owl: A library for
  $\omega$-words, automata, and ltl. In: ATVA. Lecture Notes in Computer
  Science, vol. 11138, pp. 543--550. Springer (2018)

\bibitem{Kucera05}
Ku{\v c}era, A., Strej{\v{c}}ek, J.: The stuttering principle revisited. Acta
  Informatica  \textbf{41}(7--8),  415--434 (2005)

\bibitem{Lamp83a}
{L. Lamport}: What good is temporal logic? In: {R.E.A. Mason} (ed.) {IFIP}
  Congress. pp. 657--667. North-Holland, Amsterdam (1983)

\bibitem{LarPet97}
Larsen, K.G., Pettersson, P., Yi, W.: Uppaal in a nutshell. International
  Journal on Software Tools for Technology Transfer  \textbf{1}(1-2),  134--152
  (1997)

\bibitem{LiSGFTO10}
Li, Y., Sakiyama, K., Gomisawa, S., Fukunaga, T., Takahashi, J., Ohta, K.:
  Fault sensitivity analysis. In: Mangard, S., Standaert, F.X. (eds.) CHES.
  Lecture Notes in Computer Science, vol.~6225, pp. 320--334. Springer (2010)

\bibitem{Lipp0G0HFHMKGYH18}
Lipp, M., Schwarz, M., Gruss, D., Prescher, T., Haas, W., Fogh, A., Horn, J.,
  Mangard, S., Kocher, P., Genkin, D., Yarom, Y., Hamburg, M.: Meltdown:
  Reading kernel memory from user space. In: USENIX Security Symposium. pp.
  973--990. USENIX Association (2018)

\bibitem{LorberRNA17}
Lorber, F., Rosenmann, A., Nickovic, D., Aichernig, B.K.: Bounded
  determinization of timed automata with silent transitions. Real-Time Systems
  \textbf{53}(3),  291--326 (2017)

\bibitem{NguyenKJDJ17}
Nguyen, L.V., Kapinski, J., Jin, X., Deshmukh, J.V., Johnson, T.T.:
  Hyperproperties of real-valued signals. In: MEMOCODE. pp. 104--113. ACM
  (2017)

\bibitem{Ouaknine2009}
Ouaknine, J., Rabinovich, A., Worrell, J.: Time-bounded verification. In:
  Proceedings of CONCUR 2009. LNCS, vol.~5710, pp. 496--510. Springer (2009)

\bibitem{OuaknineW04}
Ouaknine, J., Worrell, J.: On the language inclusion problem for timed
  automata: Closing a decidability gap. In: LICS. pp. 54--63. IEEE Computer
  Society (2004)

\bibitem{OuaWor07}
Ouaknine, J., Worrell, J.: On the decidability and complexity of metric
  temporal logic over finite words. Logical Methods in Computer Science
  \textbf{3}(1) (2007)

\bibitem{Pnueli1977}
Pnueli, A.: The temporal logic of programs. In: FOCS. pp. 46--57. IEEE (1977)

\bibitem{Raskin1999}
Raskin, J.F.: Logics, automata and classical theories for deciding real time.
  Ph.D. thesis, FUNDP (Belgium) (1999)

\bibitem{Roscoe95}
Roscoe, A.W.: Csp and determinism in security modelling. In: S\&P. pp.
  114--127. IEEE Computer Society (1995)

\bibitem{Sistla1985}
Sistla, A.P., Vardi, M.Y., Wolper, P.: The complementation problem for
  {B{\"u}chi} automata with applications to temporal logic (extended abstract).
  In: ICALP. LNCS, vol.~194, pp. 465--474. Springer (1985)

\bibitem{TsaiTH13}
Tsai, M.H., Tsay, Y.K., Hwang, Y.S.: Goal for games, omega-automata, and
  logics. In: CAV. Lecture Notes in Computer Science, vol.~8044, pp. 883--889.
  Springer (2013)

\bibitem{Wilke1994}
Wilke, T.: Specifying timed state sequences in powerful decidable logics and
  timed automata. In: FTRTFT. LNCS, vol.~863, pp. 694--715. Springer (1994)

\bibitem{ZdancewicM03}
Zdancewic, S., Myers, A.C.: Observational determinism for concurrent program
  security. In: CSFW. p.~29. IEEE Computer Society (2003)

\end{thebibliography}

\end{document}